\documentclass{llncs} 
\newtheorem{defn}{Definition} 
\newtheorem{Lem}[defn]{Lemma}

\newtheorem{Cor}{Corollary}
\newtheorem{thm}{Theorem}

\author{A. N. Trahtman\thanks{Email: trakht@macs.biu.ac.il,9 Dec 2005, DMTCS 9:2, 2007, 3-10 }}

\title{ The \v{C}erny Conjecture for Aperiodic Automata}

\institute{Bar-Ilan University, Department of Mathematics\\
    52900, Ramat Gan, Israel, }

\begin{document}

\maketitle

\begin{abstract}
{A word $w$ is called a \emph{synchronizing}
(\emph{recurrent, reset, directable}) word of a deterministic finite automaton
(DFA) if $w$ brings all states of the automaton to some specific state; a DFA
that has a synchronizing word is said to be \emph{synchronizable}. \v{C}ern\'y
conjectured in 1964 that every $n$-state synchronizable DFA possesses
a synchronizing word of length at most $(n-1)^2$. We consider automata with
aperiodic transition monoid (such automata are called \emph{aperiodic}).
We show that every synchronizable $n$-state aperiodic DFA has a synchronizing
word of length at most $n(n-1)/2$. Thus, for aperiodic automata 
as well as for automata accepting only star-free languages, 
the \v{C}ern\'y conjecture holds true.}
\end{abstract}

$\bf Keywords:$ deterministic finite automaton, synchronizing word, star-free language.

\section*{Introduction}

The problem of synchronization of DFA is natural and various aspects
of this problem were touched upon the literature.  We pay attention
to the problem of the existence and of the length of a synchronizing
word.

An important problem with a long story is estimating
the shortest length of a synchronizing word.  Best known
as the \v{C}ern\'y conjecture, it was proposed independently
by several authors. \v{C}ern\'y found in 1964 \cite{Ce}
an $n$-state DFA whose shortest synchronizing word was
of length $(n-1)^2$. He conjectured that this is the maximum
length of the shortest synchronizing word for any DFA
with $n$ states. The conjecture has been verified for several
partial cases \cite{AV1,Du,Ep,Ka,R1,Pi} but in general the question
still remains open. By now, this simply looking conjecture is arguably
one of the most longstanding open problems in the theory of finite
automata. The best upper bound for the length of the shortest
synchronizing word for DFA with $n$ states known so far is equal
to $(n^3-n)/6$ \cite{Fr,KRS,Pi83}. For the rich and intriguing
story of investigations in this area see \cite{Sa}.

The existence of some non-trivial subgroup in the transition
semigroup of the automaton is essential in many investigations of
the \v{C}ern\'y conjecture, see, e.g., \cite{Du,Pi}. We use another
approach and consider transition semigroups without non-trivial subgroups.
This condition distinguishes a wide class of so-called aperiodic automata
that, as shown by Sch\"{u}tzenberger~\cite {Sh}, accept precisely star-free
languages (also known as languages of star height $0$). Star-free
languages play a significant role in formal language theory.

We prove that every $n$-state aperiodic DFA with a state that is accessible
from every state of the automaton has a synchronizing word of length not
greater than $n(n-1)/2$, and therefore, for aperiodic automata as 
well as for automata accepting only star-free languages, 
the \v{C}ern\'y conjecture holds true.

In the case when the underlying graph of the aperiodic DFA is strongly connected,
this upper bound has been recently improved by Volkov who has reduced the estimation
to $ n(n+1)/6$.

\section{Preliminaries}

We consider a complete DFA $\mathcal{A}$ with the input alphabet $\Sigma$.
The transition graph of $\mathcal{A}$ is denoted by $\Gamma$ and the transition
semigroup of $\mathcal{A}$ is denoted by $S$.

Let $\bf p$ and $\bf q$ be two (not necessarily distinct) states of the
automaton $\mathcal{A}$. If there exists a path in $\mathcal{A}$ from the
state $\bf p$ to the state $\bf q$ and the transitions of the path are
consecutively labelled by $\sigma_1,\dots, \sigma_k\in\Sigma$ then
for $s=\sigma_1\dots \sigma_k$ we write  ${\bf q}={\bf p}s$.
We call a state $\bf q$ a \emph{sink} if for every state $\bf p$
of $\mathcal{A}$ there exists a word $s$ such that ${\bf p}s = \bf q$.
For a set $P$ of states and $s\in\Sigma^*$, let $Ps$ denote the set
$\{\mathbf{p}s\mid \mathbf{p}\in P\}$. A word $s \in \Sigma^+ $
is called a \emph{synchronizing word for $P$} if $|Ps|=1$, that is,
$\mathbf{p}s=\mathbf{q}s$ for all states $\mathbf{p},\mathbf{q}\in P$.
A word is said to be a \emph{synchronizing word of the automaton $\mathcal{A}$}
(\emph{of the graph $\Gamma$}) if it is synchronizing for the set of all states
of $\mathcal{A}$ (the set of all vertices of $\Gamma$).

A binary relation $\beta$ on the state set of $\mathcal{A}$ is called {\it stable\/}
if, for any pair of states ${\bf q},\bf p$ and any $\sigma \in \Sigma$,
from ${\bf q}\,\beta\,\bf p$ it follows ${\bf q}\sigma\,\beta\,{\bf p}\sigma$.
Recall that a stable equivalence relation on the state set of $\mathcal{A}$
is called a \emph{congruence} of $\mathcal{A}$. If $\rho$ is a congruence of
$\mathcal{A}$, we denote by $[\bf q]_{\rho}$ the $\rho$-class containing
the state $\bf q$. The {\it quotient} $\mathcal{A}/{\rho}$ is the automaton
with the states $[{\bf q}]_{\rho}$ and the transition function defined
by the rule $[{\bf q}]_{\rho}\sigma =[{\bf q}\sigma]_{\rho}$ for any
$\sigma \in \Sigma$.

For a word $s$ over the alphabet $\Sigma$, we denote its length by $|s|$.

\section{The graph $\Gamma^2$}

The direct square $\Gamma^2$ of the transition graph $\Gamma$ has as vertices
all pairs $({\bf p},{\bf q})$, where ${\bf p},{\bf q}$ are vertices of $\Gamma$.
The edges of the graph $\Gamma^2$ have the form
$({\bf p},{\bf q}) \to({\bf p}\sigma,{\bf q}\sigma)$
where $\sigma \in \Sigma$; such an edge is labelled by $\sigma$.

For brevity, a strongly connected component of a directed graph is
referred to as an \textbf{SCC}. An \textbf{SCC} $M$ of the graph $\Gamma^2$ is
called {\it almost minimal\/} if, for every pair $({\bf p},{\bf q})\in M$,
one has ${\bf p}\ne{\bf q}$ and, for every $\sigma \in \Sigma$ such that
${\bf p}\sigma \ne {\bf q}\sigma$, there exists a word $s\in\Sigma^*$ such that
$({\bf p}\sigma,{\bf q}\sigma)s= ({\bf p},{\bf q})$. We observe that then
$({\bf p}\sigma,{\bf q}\sigma)\in M$ by the definition of an \textbf{SCC}.
By $\Gamma(M)$ we denote the set of states that appear as components in the pairs
from the almost minimal \textbf{SCC} $M$.

If $M$ is an almost minimal \textbf{SCC}, we define the relation $\succ_M$
as the transitive closure of $M$ (where $M$ is treated as a relation on the
state set of our automaton). So ${\bf r} \succ_M \bf q$ if there exists
a sequence of states ${\bf r} = {\bf p}_1,\dots,{\bf p}_n =\bf q$
such that $n>1$ and $({\bf p}_i,{\bf p}_{i+1})\in M$ for all $i=1,\dots,n-1$.
Let $\succeq_M$ be the reflexive closure and $\rho_M$ the equivalent closure
of the relation $\succ_M$.

\begin{Lem}\label{12}
For any almost minimal \textbf{SCC} $M$, the relation $\succeq_M$ is stable
and the relation $\rho_M$ is a congruence.
\end{Lem}

\begin{proof}
Suppose ${\bf u}\,\rho_M\,{\bf v}$. Then there exists a
sequence of states
\begin{equation}
\label{sequence}
{\bf u}={\bf p}_1,\dots,{\bf p}_n={\bf v}
\end{equation}
such that for every integer $i<n$ at least one of the pairs
$({\bf p}_{i+1}, {\bf p}_{i})$, $({\bf p}_i, {\bf p}_{i+1})$ belongs
to the almost minimal \textbf{SCC} $M$. Therefore in the sequence of
states ${\bf r}s = {\bf p}_1s,\dots ,{\bf p}_ns = {\bf q}s$, for any
two distinct neighbors ${\bf p}_is, {\bf p}_{i+1}s$, the pair
$({\bf p}_is, {\bf p}_{i+1}s)$ or its dual belongs to $M$.
Hence ${\bf r}s$ $\rho_M$ ${\bf q}s$.

If ${\bf u}\succeq_M {\bf v}$, then there exists a sequence (\ref{sequence})
such that for every integer $i < n$ the pair $({\bf p}_i, {\bf p}_{i+1})$
belongs to $M$. Then either $({\bf p}_is, {\bf p}_{i+1}s) \in M$ or
${\bf p}_is = {\bf p}_{i+1}s$, and therefore, ${\bf p}_is \succeq_M {\bf p}_{i+1}s$
in any case. Hence ${\bf u}s \succeq_M {\bf v}s$.
\end{proof}

From the definition of the relation $\succ_M$ and Lemma \ref{12}, we obtain
\begin{Cor} \label{rem}
If ${\bf r} \succ_M \bf q$ and ${\bf r}s \notin \Gamma(M)$
for some word  $s$, then ${\bf r}s = {\bf q}s$.
\end{Cor}

We also observe the following obvious property:
\begin{Cor}\label{c}
Each state ${\bf p}$ from $\Gamma(M)$ belongs to a $\rho_{ M}$-class
of size at least two.
\end{Cor}

Let us present the following new formulation of a result from \cite{Ce}:
\begin{Lem}  \label {alg}
An automaton $\mathcal{A}$ with the transition graph $\Gamma$ is synchronizing
if and only if the graph $\Gamma^2$ has a sink.
\end{Lem}

\begin{proof}
Let $s$ be a synchronizing word of $\mathcal{A}$. Then the unique pair
of the set ${\Gamma^2}s$ is a sink of $\Gamma^2$. Conversely, the components
of a sink of $\Gamma^2$ obviously are equal. Let ($\bf t, t$) be a sink.
For any pair ($\bf p, q$), there exists a word $s$ such
that $({\bf p}, {\bf q})s = (\bf t, t$), that is, ${\bf p}s = {\bf q}s = \bf t$.
Some product of such words $s$ taken for all  pairs of distinct states from
$\Gamma$ is a synchronizing word of the graph $\Gamma$.
\end{proof}

\begin{Lem}\label{l1}
The sets of synchronizing words of the graphs $\Gamma$
and $\Gamma^2$  coincide.
\end{Lem}

\begin{proof}
Let $s$ be a synchronizing word of the graph $\Gamma$.
Then there is a state $\bf q$ from $\Gamma$ such that
${\bf p}s =\bf q$ for every state $\bf p$. Therefore for
every pair $(\bf p, r$) one has
$({\bf p, r})s =(\bf q, q$). Thus, $s$ is a synchronizing
word of the graph $\Gamma^2$.

Now let $t$ be a synchronizing word of the graph $\Gamma^2$.
Then there is a pair $(\bf q, v$) such that
$({\bf p, r})t =(\bf q, v$) for every pair $(\bf p, r$).
Therefore ${\bf p}t =\bf q$ for an arbitrary state $\bf p$ from $\Gamma$
and ${\bf r}t =\bf v$ for an arbitrary state $\bf r$ from $\Gamma$.
Consequently, ${\bf v} =\bf q$ and $t$ is a synchronizing
word of the graph $\Gamma$.
\end{proof}

\section{Aperiodic automata}

A semigroup without non-trivial subgroups is called {\it aperiodic}.
A DFA with  aperiodic transition semigroup is called {\it aperiodic\/} too.

Let us recall that the syntactic semigroup of a star-free language
is finite and aperiodic \cite {Sh} and the semigroup satisfies
the identity $x^n=x^{n+1}$ for some suitable $n$. Therefore,
for any state ${\bf p}\in \Gamma$, any $s \in S$ and for some
suitable $k$, one has ${\bf p}s^k={\bf p}s^{k+1}$.

\begin{Lem}\label{l10}
Let $\mathcal{A}$ be an aperiodic DFA. Then the existence
of a sink in $\mathcal{A}$ is equivalent to the existence
of a synchronizing word.
\end{Lem}

\begin{proof}
It is clear that, for any DFA, the existence of a
synchronizing word implies the existence of a sink.

Now suppose that $\mathcal{A}$ has at least one sink.
For any state $\bf p$ and any sink $\bf p_0$, there exists
an element $s$ from the transition semigroup $S$ such that ${\bf p}s= \bf p_0$.
The semigroup $S$ is aperiodic, whence for some positive integer $m$
we have $s^m = s^{m+1}$. Therefore ${\bf p}s^m={\bf p}s^{m+1}={\bf p_0}s^m$,
whence the element $s^m$ brings both $\bf p$ and $\bf p_0$ to
the same state ${\bf p_0}s^m$ which is a sink again. We repeat
the process reducing the number of states on each step.
Then some product of all elements of the form $s^m$ arising
on each step brings all states of the automaton to some sink.
Thus, we obtain in this way a synchronizing word.
\end{proof}

Let $M$ be an almost minimal \textbf{SCC}.
A $t$-\emph{cycle} for $M$ is a sequence of states
\begin{equation}
\label{cycle}
{\bf p}_1,{\bf p}_2,\dots,{\bf p}_{m-1},{\bf p}_m={\bf p}_1
\end{equation}
such that $n>1$ and $({\bf p}_i,{\bf p}_{i+1})\in M$ for all $i=1,\dots,m-1$.
The next observation is the key ingredient of the proof.
\begin{Lem}\label {l6}
Let $\mathcal{A}$ be an aperiodic DFA and let $M$ be an almost minimal \textbf{SCC}.
Then there is no $t$-cycle for $M$ and the quasi-order $\succeq_M$ is a partial order.
 \end{Lem}

\begin{proof}
Suppose that (\ref{cycle}) is a $t$-cycle of minimum size $m$ among all $t$-cycles
for the almost minimal \textbf{SCC} $M$. Let us first establish that $m>2$. Indeed,
${\bf p}_1 \ne {\bf p}_2$, whence $m>1$. If $m=2$ then the two pairs $({\bf p}_1, {\bf p}_2$)
and $({\bf p}_2, {\bf p}_1$) belong to the \textbf{SCC} $M$.
For some element $u$ from the transition semigroup $S$, we have
$({\bf p}_1, {\bf p}_2)u =({\bf p}_2, {\bf p}_1$). Therefore
${\bf p}_1u= {\bf p}_2$, ${\bf p}_2u= {\bf p}_1$,
whence ${\bf p}_1u^2= {\bf p}_1 \ne {\bf p}_1u$. This implies
${\bf p}_1u^{2k}= {\bf p}_1 \ne {\bf p}_1u={\bf p}_1u^{2k+1}$
for any integer $k$. However, semigroup $S$ is finite and aperiodic,
and therefore, for some $k$ we have $u^{2k}= u^{2k+1}$,
whence ${\bf p}_1u^{2k}= {\bf p}_1u^{2k+1}$, a contradiction.

Thus, we can assume that $m>2$ and suppose that the states ${\bf
p}_1, {\bf p}_2, {\bf p}_3$ are distinct. For some element $s \in
S$, we have
$({\bf p}_1, {\bf p}_2)s = ({\bf p}_2, {\bf p}_3$). Hence
$${\bf p}_2= {\bf p}_1s,\ {\bf p}_3 = {\bf p}_1s^2.$$

For any element $u \in S$ and any pair $({\bf p}_i, {\bf p}_{i+1})$
from $M$, we have either ${\bf p}_iu={\bf p}_{i+1}u$ or
$({\bf p}_iu,{\bf p}_{i+1}u)\in M$. Therefore, for any element $u \in S$,
the sequence of states ${\bf p}_1u,\dots ,{\bf p}_mu$ either reduces to
just one element repeated $m$ times or forms a $t$-cycle
of size $m$ (because of the minimality of $m$).

The states ${\bf p}_1, {\bf p}_1s, {\bf p}_1s^2$ are distinct.
Since $S$ is an aperiodic finite semigroup, there exists some integer $\ell$
such that $s^{\ell} \ne s^{\ell+1} =s^{\ell+2}$. Therefore there exists an $k \le\ell$
such that  ${\bf p}_1 s^{k} \ne {\bf p}_1s^{k+1} = {\bf p}_1s^{k+2}$ the
sequence ${\bf p}_1s^k$, ${\bf p}_2s^{k}={\bf p}_1s^{k+1}$,
 ${\bf p}_3s^{k}={\bf p}_1s^{k+2}$,\dots , ${\bf p}_ms^k$
has more than 1 but less than $m$ distinct elements. This
contradicts the conclusion of the previous paragraph applied
to the element $u=s^k$.

It is easy to see that if the quasi-order $\succeq_M$ is not antisymmetric
that there exists a $t$-cycle for $M$. Hence $\succeq_M$ is a partial order.
\end{proof}

\section{The \v{C}ern\'y conjecture}

\begin{Lem}\label{18}
Let $\mathcal{A}$ be an aperiodic DFA with $n$ states and strongly
connected graph, $M$ an almost minimal \textbf{SCC}. Let $r$ be the number
of $\rho_M$-classes and let $R$ be a $\rho_M$-class. Then $|Rs|=1$
for some word $s\in \Sigma^*$ of length ar most $(n-r+1)(n-1)/2$.
\end{Lem}

\begin{proof}
Suppose $|R|>1$. Let $Max$ be the set of all maximal and $Min$
the set of all minimal states from $R$ with respect to the order
$\succ_M$. Observe that there is no ambiguity here: since the order
$\succ_M$ is contained in the congruence $\rho_M$, maximal (minimal)
states of the ordered set $(R,\succ_M)$ are precisely those maximal
(minimal) states of the automaton $\mathcal{A}$ that belong to $R$.
Further, $Max\cap Min = \emptyset$ because the congruence $\rho_M$
is the equivalent closure of the order $\succ_M$ whence for every
state ${\bf q}\in R$ there must be a state ${\bf p}\in R$ such that
either ${\bf q}\succ_M{\bf p}$ or  ${\bf p}\succ_M{\bf q}$. Without
loss of generality, we may assume that $|Max|\ge|Min|$. Then
$|Min| \le |R|/2\le (n-r+1)/2$.

We need three properties of ordered sets of the form $(Rs,\succ_M)$
where $s$ is a word. Let $Max_s$ ($Min_s$) stand for the set of all
maximal (respectively all minimal) states in $(Rs,\succ_M)$.

\medskip

\noindent\textbf{\emph{Claim 1.}} If $|Rs|>1$, then $Min_s\cap Max_s=\emptyset$.

\begin{proof}
Take an arbitrary state ${\bf p'}\in Min_s$ and let ${\bf q'}$ be any state
in $Rs\setminus\{{\bf p'}\}$. Consider some preimages ${\bf p},{\bf q}\in R$
of ${\bf p'}$ and respectively ${\bf q'}$. Since $R$ is a $\rho_M$-class,
there is a sequence of states ${\bf q}_0,{\bf q}_1,\dots,{\bf q}_k\in R$ such
that ${\bf p}={\bf q}_0$, ${\bf q}_k={\bf q}$, and for each $i=1,\dots,k$ either
${\bf q}_{i-1}\succeq_M{\bf q}_i$ or ${\bf q}_i\succeq_M{\bf q}_{i-1}$.
Let ${\bf q'}_i={\bf q}_is\in Rs$, $i=0,\dots,k$. Since the order $\succeq_M$
is stable (Lemma \ref{12}), we conclude that there is a sequence of states
${\bf q'}_0,{\bf q'}_1,\dots,{\bf q'}_k\in Rs$ such
that ${\bf p'}={\bf q'}_0$, ${\bf q'}_k={\bf q'}$, and for each $i=1,\dots,k$
either ${\bf q'}_{i-1}\succeq_M{\bf q'}_i$ or ${\bf q'}_i\succeq_M{\bf q'}_{i-1}$.
Since ${\bf p'}\ne{\bf q'}$, some of these inequalities must be strict.
Let $j$ be the least index such that ${\bf q'}_{j-1}\ne{\bf q'}_j$.
Then ${\bf p'}={\bf q'}_0=\dots={\bf q'}_{j-1}$ whence
either ${\bf p'}\succ_M{\bf q'}_j$ or ${\bf q'}_j\succ_M{\bf p'}$.
As the first inequality would contradict the assumption that ${\bf p'}$
is a minimal element of $(Rs,\succ_M)$, we conclude that the second
inequality holds true whence ${\bf p'}$ is not a maximal element of
$(Rs,\succ_M)$. Thus, no state in $Min_s$ can belong to $Max_s$.
\end{proof}

\noindent\textbf{\emph{Claim 2.}} If $s,t\in\Sigma^*$ are two arbitrary
words, then $Min_{ts}\subseteq Min_t\,s$.

\begin{proof}
Take an arbitrary state ${\bf p'}\in Min_{ts}$ and consider its arbitrary
preimage ${\bf p}\in Rt$. There exists a state ${\bf q}\in Min_t$ such that
${\bf p}\succeq_M{\bf q}$. Since the order $\succeq_M$ is stable (Lemma \ref{12}),
we then have ${\bf p'}={\bf p}\,s\succeq_M{\bf q}\,s={\bf q'}$. The state ${\bf q'}$
belongs to the set $Rts$, and therefore, ${\bf q}'={\bf p'}$ because ${\bf p'}$ has
been chosen to be a minimal element in this set. Thus, we have found a preimage
for ${\bf p'}$ in $Min_t$ whence $Min_{ts}\subseteq Min_t\,s$.
\end{proof}

\noindent\textbf{\emph{Claim 3.}} For any word $t\in\Sigma^*$, there exists
a word $s\in\Sigma^*$ of length at most $n-1$ such that either $|Rst|=1$ or
$|Min_{ts}|<|Min_t|$.

\begin{proof}
Now take arbitrary state ${\bf q}\in Min_t$. Since the graph $\Gamma$ is
strongly connected, there exists a word that maps ${\bf q}$ to an element
of $Max$. If $s$ is a word of minimum length with this property then the path
labelled $s$ does not visit any state of $\mathcal{A}$ twice whence $|s|\le n-1$.
Observe that since ${\bf q}s\in Max$, we also have ${\bf q}s\in Max_{ts}$. Therefore
either $|Rst|=1$ or, by Claim~1 above, ${\bf q}s\notin Min_{ts}$. Since
${\bf q}s\in Min_t\,s$, we conclude from Claim~2 that $Min_{ts}\subset Min_t\,s$.
Thus, $|Min_{ts}|<|Min_t\,s|\le|Min_t|$.
\end{proof}

Using Claim~3, we can easily complete the proof of the lemma. Indeed, applying
it to the case when $t$ is the empty word, we can find a word $s_1$ of length
at most $n-1$ such that either $|Rs_1|=1$ or $|Min_{s_1}|<|Min|$. In the latter
case, applying Claim~3 again, we can find a word $s_2$ of length at most $n-1$
such that either $|Rs_1s_2|=1$ or $|Min_{s_1s_2}|<|Min_{s_1}|$, and so on. Clearly,
the process will stop after at most $|Min|$ steps yielding a word $s=s_1s_2\cdots s_k$
(with $k\le|Min|$ and $|s_i|\le n-1$ for each $i=1,\dots,k$) such that $|Rs|=1$.
Since $|Min|\le(n-r+1)/2$, we have $|s|\le(n-r+1)(n-1)/2$ as required.
\end{proof}

\begin{thm}  \label{t2}
If the transition graph $\Gamma$ of an aperiodic DFA $\mathcal{A}$
with $n$ states is strongly connected, then $\mathcal{A}$ has a synchronizing
word of length at most $(n-1)n/2$.
 \end{thm}

\begin{proof}
All states of a DFA whose transition graph is strongly connected
are sinks. Therefore the automaton is synchronizable (Lemma \ref{l10}).

There exists at least one almost minimal \textbf{SCC} $M$ in $\Gamma^2$
because the number of \textbf{SCC}'s is finite and the set of \textbf{SCC}'s
is partially ordered under the attainability relation. Consider the congruence
$\rho_M$ (Lemma \ref{12}) and the quotient $\Gamma/\rho_M$.

It is clear that any synchronizing word of $\Gamma$ synchronizes also $\Gamma/\rho_M$
and that $\Gamma/\rho_M$ is aperiodic and strongly connected. Therefore the graph $\Gamma$ has
a synchronizing word $uv$ where $u$ is a synchronizing word of $\Gamma/\rho_M$ and $v$
is a synchronizing word of the preimage $R$ of the singleton set $(\Gamma/\rho_M)\,u$.
By Corollary~\ref{c}, $\rho_M$ is not trivial, therefore $r=|\Gamma/\rho_M|< n$
and we can use induction assuming $|u|\le (r-1)r/2$. By Lemma \ref{18},
a word $v$ of length at most $(n-r+1)(n-1)/2$ synchronizes $R$. Therefore
$|uv| \le \frac{r(r-1)}2 + \frac{(n-r+1)(n-1)}2
\le \frac{(r-1)(n-1)}2+ \frac{(n-r+1)(n-1)}2
\le \frac{n(n-1)}2$
as required.
\end{proof}

Let us go to the general case.
 \begin{thm}  \label{g1}
Let $\mathcal{A}$ be an aperiodic DFA with $n$ states.
Then the existence of a sink in $\mathcal{A}$ is equivalent
to the existence of a synchronizing word of length at most $n(n-1)/2$.
\end{thm}

\begin{proof}
It is clear that the existence of a synchronizing word implies
the existence of a sink.

For the converse, let us consider a DFA with at least one sink.
By Lemma \ref{l10}, the automaton is synchronizable. We may assume in
view of Theorem \ref{t2} that the transition graph $\Gamma$ of $\mathcal{A}$
is not strongly connected.

It is clear that the collection $C$ of all sinks of $\Gamma$ forms
a \textbf{SCC} of $\Gamma$ which is the least \textbf{SCC} with respect to the
attainability order. Let $r<n$ stand for the cardinality of $C$. Let $\Gamma_i$
($i =1, 2, \dots , k$) be all other \textbf{SCC}'s of $\Gamma$. We may assume
that $\Gamma_i$ are numbered so that $i\le j$ whenever there is a path in $\Gamma$
from $\Gamma_i$ to $\Gamma_j$. Let $r_i$ be the cardinality of $\Gamma_i$. It easily
follows from \cite[Theorem 6.1]{R2} that there exists a word $s_i$ of length at most
$r_i(r_i+1)/2$ such that $\Gamma_is_i \cap \Gamma_i$ is empty. Then the product
$s_1\cdots s_k$ of maps $\Gamma$ into the \textbf{SCC} $C$. By Theorem \ref{t2},
$C$ has a synchronizing word $s$ of length at most $r(r-1)/2$. Therefore the
word $s_1\cdots s_ks$ synchronizes $\Gamma$ and
\begin{equation}
\label{sum}
|s_1\dots s_ks| \le \sum_{i=1}^k \frac{r_i(r_i+1)}2 + \frac{r(r-1)}2.
\end{equation}
Using the equality $\sum_{i=1}^k r_i+r=n$, it is easy to calculate that the right-hand
side of the inequality \ref{sum} does not exceed $(n-1)n/2$.
\end{proof}

\begin{Cor}
The \v{C}ern\'y conjecture holds for aperiodic automata.
 \end{Cor}

\section*{Acknowledgments}
I am very grateful to  M. V. Volkov for helpful and detailed comments
that proved to be highly useful in improving the presentation and style
of the paper.

\end{document}